%% file: ms.tex
\documentclass[letterpaper, 10 pt, conference,usenames,dvipsnames,final]{ieeeconf}  

\IEEEoverridecommandlockouts                              
\let\labelindent\relax

\title{Graph-based coherence identification and frequency-domain area aggregation\\ for dynamic model reduction in large-scale networks}

\title{Spectral coherence clustering and model reduction in large-scale\\ network systems}

\title{Spectral clustering and model reduction for weakly-connected coherent network systems}

\author{Hancheng Min and Enrique Mallada\thanks{H. Min and E. Mallada are with the Department of Electrical and Computer Engineering, Johns Hopkins University, Baltimore, MD 21218, USA{\tt\small \{hanchmin, mallada\}@jhu.edu}.}
\thanks{This work was supported by the NSF HDR
TRIPODS Institute for the Foundations of Graph and Deep Learning (NSF grant 1934979),
the NSF AMPS Program (NSF grant 1736448), and the NSF CAREER Program (NSF grant
1752362). The author thanks Professor Steven Low for the insightful discussion that inspired this work.}}
\usepackage{mysty}
\usepackage{bbold,amsmath,amsthm,amsfonts,amssymb}
\usepackage{subfiles}
\allowdisplaybreaks
\usepackage{environ}
\usepackage{cite}
\usepackage{accents}
\usepackage{balance}
\usepackage{ifthen}
\usepackage[ruled,vlined]{algorithm2e}
\usepackage{verbatim}
\input{edits}
\usepackage{graphicx}
\usepackage{subcaption}

\newboolean{archive}
\setboolean{archive}{true}

\usepackage{tikz}
\usetikzlibrary{shapes, arrows}
\usepackage{enumitem}
\tikzstyle{block} = [draw, rectangle, minimum height=3.5em, minimum width=3.5em]
\tikzstyle{sum} = [draw, circle, node distance=1cm]
\tikzstyle{input} = [coordinate] \tikzstyle{output} = [coordinate]
\tikzstyle{tmp} = [coordinate]

\newif\ifshownotes
\shownotestrue
\definecolor{notetext}{rgb}{0.7,0,0}

\def\thesubsubsection{\arabic{section}}
\newtheorem{theorem}{Theorem}
\newtheorem{lemma}[theorem]{Lemma}

\newtheorem{proposition}[theorem]{Proposition}
\newtheorem*{claim}{Claim}

\theoremstyle{remark}
\newtheorem*{rem}{Remark}

\begin{document}
\maketitle
\thispagestyle{plain}
\pagestyle{plain}
\begin{abstract}
    We propose a novel model-reduction methodology for large-scale dynamic networks with tightly-connected components. First, the coherent groups are identified by a spectral clustering algorithm on the graph Laplacian matrix that models the network feedback. Then, a reduced network is built, where each node represents the aggregate dynamics of each coherent group, and the reduced network captures the dynamic coupling between the groups. Our approach is theoretically justified under a random graph setting. 
    Finally, numerical experiments align with and validate our theoretical findings.
\end{abstract}
\section{Introduction}
    In networked dynamical systems, coherence refers to a coordinated behavior from a group of nodes such that all nodes have similar dynamical responses to some external disturbances. Coherence analysis is useful in understanding the collective behavior of large networks including consensus networks~\cite{Olfati-Saber20041520}, transportation networks~\cite{Bamieh2012}, and power networks~\cite{chow1982time}. However, little we know about the underlying mechanism that causes such a coherent behavior in various networks.

    Classic slow coherence analyses~\cite{chow1982time,ramaswamy1996,romeres2013,tyuryukanov2021,fritzsch2022} (with applications mostly to power networks) usually consider the second-order electro-mechanical model without damping: $\ddot{x}=-M^{-1}Lx$, where $M$ is the diagonal matrix of machine inertias, and $L$ is the Laplacian matrix whose elements are synchronizing coefficients between pair of machines. The coherency or synchrony~\cite{ramaswamy1996} (a generalized notion of coherency) is identified by studying the first few slowest eigenmodes (eigenvectors with small eigenvalues) of $M^{-1}L$, the analysis can be carried over to the case of uniform~\cite{chow1982time} and non-uniform~\cite{romeres2013} damping. However, such state-space-based analysis is limited to very specific node dynamics (second order) and do not account for, more complex dynamics or controllers that are usually present at a node level; e.g., in the power systems literature~\cite{jpm2021tac, jbvm2021lcss, ekomwenrenren2021}.  
    Therefore, we need a coherence identification procedure that works for more general node dynamics.
    
    Recently, it has been theoretically established that coherence naturally emerges when a group of nodes are tightly-connected, regardless of the node dynamics, as long as the interconnection remains stable~\cite{min2019cdc,min2021a}. The analysis also provides an asymptotically (as the network connectivity increases) exact characterization of the coherent response, which amounts an harmonic sum of individual node transfer functions. Thus, in a sense, coherence identification is closely related to the problem of finding tightly connected components in the network, for which many clustering algorithms based on the spectral embedding of graph adjacency or Laplacian matrix, are proposed and theoretically justified~\cite{Bach2003}.
    
    This leads to the natural question: Can these graph-based clustering algorithms be adopted for coherence identification in networked dynamical systems? Intuitively, when we apply those clustering algorithms to identify tightly-connected components in the network, each component should be coherent also in the dynamical sense. Then, applying ~\cite{min2019cdc,min2021a} for each cluster, should lead to a good model for each coherent group, which after interconnected with an appropriately chosen reduced graph should lead to a good network-reduced aggregate model of the dynamic interactions among across coherent components.

    In this paper, we formalize and theoretically justify this seemingly na\"ive approach utilizing the recent frequency-domain analysis for coherence~\cite{min2021a} and dynamics aggregation~\cite{min2020lcss}. Specifically, we prose a novel approximation model for large-scale networks with two tightly-connected components/groups. The model is constructed in two stages: First, the coherent groups are identified by a spectral clustering algorithm solely on the graph Laplacian matrix of the network; Then a two-node network, in which each node represents the aggregate dynamics of one coherent group, approximates the dynamical interactions between the two coherent groups in the original network. We show that our algorithm achieves perfect clustering for coherence identification and has good accuracy in modeling the inter-group dynamical interaction in the network with high probability when the network graph is randomly generated from a weight stochastic block model. Lastly, we apply our algorithm to modeling the frequency response in power networks with IEEE 68-bus test system, and the numerical results align with our theoretical findings.

    Unlike previous coherence analysis~\cite{chow1982time,ramaswamy1996,romeres2013}, our approach is dynamic-agnostic in that the coherence is identified solely by network connections, which works for the case when nodes are equipped with complicated controllers. Moreover, our model is suitable for the control design aiming mostly at response shaping~\cite{jbvm2021lcss} as the proposed two-node model clearly shows how implemented controller would affect the aggregate dynamics and the inter-group interaction.

    The rest of the paper is organized as follows: We formalize the coherence identification problem in Section \ref{sec_prelim} and also introduce the spectral clustering algorithm. Then we propose our approximation model in Section \ref{sec_ideal_net_model} and provide theoretical justification in Section \ref{sec_random_net_model}. Lastly, we validate our model by several numerical experiments in Section \ref{sec_num}
    
    \emph{Notation:}~For a vector $x$, $\|x\|=\sqrt{x^Tx}$ denotes the $2$-norm of $x$, $[x]_i$ denotes its $i$-th entry, and for a matrix $A$, $\|A\|$ denotes the spectral norm. We let $I_n$ denote the identity matrix of order $n$, $V^T$ denote the conjugate transpose of matrix $V$, $\one_n$ denote $[1,\cdots,1]^T $ with dimension $n$, and  $[n]$ denote the set $\{1,2,\cdots,n\}$. For non-negative random variables $X(n),Y(n)$, ordering, we write $X(n)\sim\mathcal{O}_p(Y(n))$ if $\exists M>0$, s.t. $\lim_{n\ra \infty}\prob\lp X(n)\leq M Y(n)\rp=1$. We write $f(n)\sim\Omega_p(g(n))$ if $\exists M>0$, s.t. $\lim_{n\ra \infty}\prob\lp X(n)\geq M Y(n)\rp=1$.


\section{Preliminaries}\label{sec_prelim}
\subsection{Network Model}
Consider a network consisting of $n$ nodes ($n\geq 2$), indexed by $i\in[n]$ with the block diagram structure in Fig.\ref{fig_blk_digm}. $L$ is the Laplacian matrix of an undirected, weighted graph that describes the network interconnection. We further use $f(s)$ to denote the transfer function representing the dynamics of the network coupling, and $G(s)=\mathrm{diag}\{g_i(s)\}$ to denote the nodal dynamics, with $g_i(s),\ i\in[n]$, being an SISO transfer function representing the dynamics of node $i$. 
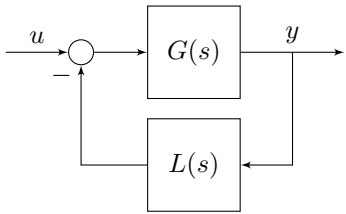
\begin{figure}[ht]
	\centering
	\begin{tikzpicture}[auto, node distance=1.5cm,>=latex']
	\node [input, name=input] {};
	\node [sum, right of=input] (sum) {};
	\node [block, right of=sum] (plant) {$G(s)$};
	\node [output, right of=plant, node distance=2cm] (output) {};
	\node [block, below of=plant] (laplacian) {$L(s)$};
	\draw [draw,->] (input) -- node {$u$} (sum);
	\draw [->] (sum) -- (plant);
	\draw [->] (plant) -- node [name=y]{$y$} (output);
	\draw [->] (y) |- (laplacian);
	\draw [->] (laplacian) -| node[pos=0.95]{$-$}(sum);
		
	\end{tikzpicture}
	\caption{Block Diagram of General Networked Dynamical Systems}\label{fig_blk_digm}
\end{figure}
The network takes a vector signal $u=[u_1,\cdots,u_n]^T$ as input, whose component $u_i$ is the disturbance or input to node $i$. The network output $y=[y_1,\cdots,y_n]^T$ contains the individual node outputs $y_i,i=1,\cdots,n$. We are interested in characterizing and approximating the response of the transfer matrix $T_{yu}(s)$ under certain assumptions on the network topology, i.e., the Laplacian matrix $L$.

Many existing networks can be represented by this structure. For example, for the first-order consensus network~\cite{Olfati-Saber20041520}, $f(s)=1$, and the node dynamics are given by $g_i(s)=\frac{1}{s}$. For power networks~\cite{Paganini2019tac}, $f(s)=\frac{1}{s}$, $g_i(s)$ are the dynamics of the generators, and $L$ is the Laplacian  matrix representing the sensitivity of power injection w.r.t. bus phase angles. Finally, in transportation networks~\cite{Jadbabaie2003988}, $g_i(s)$ represent the vehicle dynamics whereas $f(s)L$ describes local inter-vehicle information transfer. 

Recent work~\cite{min2019cdc,min2021a} has shown that, under mild assumptions, the following holds\footnote{In~\cite{min2021a}, the transfer matrix $\frac{1}{n}\bar{g}(s)\one\one^T$ appeared in the limit, where $\bar{g}(s)=\lp\frac{1}{n}\sum_{i=1}^ng_i^{-1}(s)\rp^{-1}$. It is easy to verify that $\frac{1}{n}\bar{g}(s)\one\one^T=\hat{g}(s_0)\one\one^T$ } for almost any $s_0\in\mathbb{C}$,
\be
    \lim_{\lambda_2(L)\ra\infty}\|T_{yu}(s_0)-\hat{g}(s_0)\one\one^T\|=0\,,\label{eq_coherent}
\ee
where
\be
    \hat{g}(s)=\lp \sum_{i=1}^ng_i^{-1}(s)\rp^{-1}\,.\label{eq_hat_g_def}
\ee
That is, when the algebraic connectivity $\lambda_2(L)$ of the network is high, one can approximate $T_{yu}(s)$ by a rank-one transfer matrix. Such a rank-one transfer matrix $\hat{g}(s_0)\one\one^T$ precisely describes the coherent behavior of the network: The network takes the aggregated input $\hat{u}=\one^Tu=\sum_{i=1}^nu_i$, and responds coherently as $\hat{y}\one$, where $\hat{y}=\hat{g}(s)\hat{u}$. Therefore, it suffices to study $\hat{g}(s)$ to understand the coherent behavior in a tightly-connected network.

However, practical networks are not necessarily tightly-connected. Instead, they often contain multiple groups of nodes such that within each group, the nodes are tightly-connected while between groups, the nodes are weakly-connected. Then the network dynamics can be reduced to dynamic interactions among these groups. In order to approximate such interaction, it is natural to, first identify \emph{coherent groups}, or \emph{coherent areas}, in the network, then apply the aforementioned analysis to obtain the coherent dynamics $\hat{g}(s)$ for each group, and replace the entire coherent group by an aggregate node with $\hat{g}(s)$. But the question remains as to how one would identify coherent groups to start with, and how should we model the interaction among aggregate nodes such that it approximates the interaction among coherent groups in the original network. We start with the problem of identifying coherent groups.

\subsection{Spectral Clustering}
Spectral clustering\cite{bach2004learning} is a popular technique for identifying tightly-connected components in a network. Algorithm \ref{algo_sc} describes its simplest form for identifying two groups in a network based on the graph Laplacian matrix $L$. The algorithm computes the eigenvector $v_2(L)$ of $L$ associated with the second smallest eigenvalue $\lambda_2(L)$, and group the nodes based on the sign of entries of $[v_2(L)]_i,i=1,\cdots,n$, so that nodes with non-negative $[v_2(L)]_i$ are in one group, and others in another group.
\begin{algorithm}
    \KwData{Symmetric Laplacian Matrix $L$}
    Find $\lambda_2(L),v_2(L)$ from the eigendecomposition of $L$ \\
    $\mathcal{I}_a\la \emptyset, \mathcal{I}_b\la \emptyset$ \\
    \For{$i=1,\cdots,n$}{
    \uIf{$[v_2(L)]_i\geq 0$}{$\mathcal{I}_a\la \mathcal{I}_a\cup \{i\}$}
    \Else{$\mathcal{I}_b\la \mathcal{I}_b\cup \{i\}$}
    }
    \KwResult{$\mathcal{I}_a,\mathcal{I}_b$,$\lambda_2(L)$}
    \caption{Spectral Clustering for detecting two communities}
    \label{algo_sc}
\end{algorithm}

If the network connection has a block structure, then such a simple algorithm performs well. More precisely, suppose the Laplacian matrix is of the following form:
\begin{align}
    &\;L_\text{blk}=D-A_{\text{blk}},\label{eq_L_blk_def}\\
    &\;D=\dg\{A_\text{blk}\one\},\ A_\text{blk}=\bmt\alpha \one_{n_a}\one_{n_a}^T & \beta \one_{n_a}\one_{n_b}^T\\
    \beta \one_{n_b}\one_{n_a}^T&\alpha \one_{n_b}\one_{n_b}^T\emt\,,\nonumber
\end{align}
where, $n_a+n_b=n$, and $0\leq \beta \leq \alpha$. Then such a network has, by construction, two coherent groups: one consisting of the first $n_a$ nodes and another consisting of the remaining $n_b$ nodes. From the adjacency matrix $A_\text{blk}$, it is clear that the nodes in the same coherent group are tightly connected while the nodes from different group are relatively weakly connected (since $\beta\leq \alpha$). One can show that $$v_2(L_\text{blk})=\frac{1}{\sqrt{n}}\bmt \sqrt{\frac{n_b}{n_a}}\one_{n_a}\\ -\sqrt{\frac{n_a}{n_b}}\one_{n_b}\emt\,,$$
from which Algorithm \ref{algo_sc} groups first $n_a$ nodes into one group and the rest $n_b$ into the other group. 


Obviously, one would not expect such a densely connected network in practice, then how does spectral clustering remain effective? Previous work studied the spectral clustering algorithm (and its variants) on random graphs generated from the \emph{Stochastic Block Model}~\cite{lyzinski2014perfect}, where, in its simplest form with two communities/groups, the edge between every two nodes appears in the network independently with some probability, such that intra-group edges appear more often than inter-group edges. This randomly generated adjacency matrix $A$ has expected value of the form $A_{\text{blk}}$ in \eqref{eq_L_blk_def}, and more interestingly, for large networks, $\|L_A-L_{\text{blk}}\|$ is small with high probability, where $L_A$ is the Laplacian matrix constructed from $A$. Therefore spectral clustering on $L_A$ should not be much different from one on $L_{\text{blk}}$. Indeed, it can be shown that the angle between two vector $v_2(L_A),v_2(L_{\text{blk}})$ is small with high probability such that $[v_2(L_A)]_i$ has the same sign as $[v_2(L_{\text{blk}})]_i$, which suggests that the algorithm still performs well. Thus, if one views a real network as one instance of random graphs from the stochastic block model, then spectral clustering should perform well for identifying coherent groups.

In our setting, once coherent groups are identified, one still needs to model the dynamic interaction between the two groups. To address this challenge, we will keep the same rationale used to justify spectral clustering: We first show how the interaction can be modeled under an ideal network defined as in \eqref{eq_L_blk_def} (Section \ref{sec_ideal_net_model}), and then argue that the proposed model works for random graphs as long as they remain close to its expected value with high probability (Section \ref{sec_random_net_model}).

\section{Model Reduction for Networks with Block Structure}\label{sec_ideal_net_model}
Recall that, as shown in \eqref{eq_coherent}, when $\lambda_2(L)$ is large, the network transfer matrix $T_{yu}(s)$ can be approximated by a rank-one transfer matrix. In this section, we show that when $\lambda_3(L)$ is large, the network transfer matrix can be approximated by a rank-two transfer matrix, and under an ideal two-blocks network assumption, such a transfer matrix is precisely characterized by a network of two aggregate nodes.
\subsection{Rank-two Approximation of $T_{yu}(s)$} 
Given eigendecomposition $L=\sum_{i=1}^n\lambda_i(L)v_i(L)v_i^T(L)$, we first define
$$
    T_2(s)=\bmt \frac{\one}{\sqrt{n}} & v_2(L)\emt H_2^{-1}\bmt \frac{\one^T}{\sqrt{n}} \\ v_2^T(L)\emt\,,
$$
where
\begin{align}
    &\;H_2(s)=\nonumber\\
    &\;\bmt \frac{1}{n}\one^TG^{-1}(s)\one & \frac{\one^T}{\sqrt{n}}G^{-1}(s)v_2(L)\\
    v_2^T(L)G^{-1}(s)\frac{\one}{\sqrt{n}}& v_2^T(L)G^{-1}(s)v_2(L)+\lambda_2(L)f(s)\emt\,.\label{eq_H_2_gen}
\end{align}
Then our main result is the following:
\begin{theorem}\label{thm_T2}
    For $s_0\in\compl$ that is not a pole of $f(s)$ and has these two quantities $$\|T_2(s_0)\|:= M_1, \text{and }\max_{1\leq i\leq n}|g_i^{-1}(s_0)|:= M_2\,,$$ finite. Then whenever $|f(s_0)|\lambda_3(L)\geq M_2+M_1M_2^2$, the following inequality holds:
    \be
        \lV T_{yu}(s_0)-T_2(s_0)\rV\leq \frac{\lp M_1M_2+1\rp^2}{|f(s_0)|\lambda_3(L)-M_2-M_1M_2^2}\,.\label{eq_T_norm_bd}
    \ee
\end{theorem}
The proof is shown in Appendix. The theorem states that for almost any $s_0\in\mathbb{C}$, except for poles of $T_2(s)$, zeros of $g_i(s)$ and pole of $f(s)$, one can approximate $T_{yu}(s)$ by a rank-two transfer matrix $T_2(s)$, in frequency domain. While establishing the relations between $T_{yu}(s)$ and $T_2(s)$ regarding the time-domain response is left as future research, this theorem suggests that when the network has large $\lambda_3(L)$, the network dynamics can be potentially understood by studying $T_2(s)$.

\subsection{Preliminary Case: Dense Graph with Two-blocks Structure}
While studying $T_2(s)$ itself can be interesting, we will show that under certain assumptions on the network topology, $T_2(s)$ has an even simpler and more interpretable form.

Let us thus assume the network has the Laplacian matrix as in \eqref{eq_L_blk_def}:
\begin{align*}
    &\;L_\text{blk}=D-A_{\text{blk}},\\
    &\;D=\dg\{A_\text{blk}\one\},\ A_\text{blk}=\bmt\alpha \one_{n_a}\one_{n_a}^T & \beta \one_{n_a}\one_{n_b}^T\\
    \beta \one_{n_b}\one_{n_a}^T&\alpha \one_{n_b}\one_{n_b}^T\emt\,,\nonumber
\end{align*}
where, $n_a+n_b=n$, and $0\leq \beta \leq \alpha$. 
Without loss of generality, we assume $n_a\geq n_b$. We starts with the following statement regarding the eigenvalues and eigenvectors of $L_\text{blk}$:
\begin{claim}
    For the Laplacian matrix $L_\text{blk}$ defined in \eqref{eq_L_blk_def} with $\alpha\geq \beta$ and $n_a\geq n_b$, we have
    \begin{enumerate}
        \item $\lambda_2(L_\text{blk})=n\beta,v_2(L_\text{blk})=\frac{1}{\sqrt{n}}\bmt \sqrt{\frac{n_b}{n_a}}\one_{n_a}\\ -\sqrt{\frac{n_a}{n_b}}\one_{n_b}\emt$;
        \item $\lambda_3(L_\text{blk})=n_b\alpha+n_a\beta$.
    \end{enumerate}
\end{claim}

We have shown that $T_2(s)$ approximate $T_{yu}(s)$ well if $\lambda_3(L)$ is large. Given weak inter-area connectivity, namely a small $\beta$, $\lambda_3(L_\text{blk})$ is large when 1) $\alpha$ , the intra-area connection, is large; and 2) $n_b$ is not too small, i.e., the two coherent groups has balanced size. The two conditions are reasonable: we require the former for the coherence to emerge in the first place, and the later excludes the case where the network are dominated by one large coherent group.

More importantly, the eigenvector $v_2(L_\text{blk})$, which is used to define $T_2(s)$, has an simple expression. In this case, the symmetric $H_2(s)$ in \eqref{eq_H_2_gen} can be written as
\begin{align}
    &\;H_2(s)=\nonumber\\
    &\;\bmt \frac{1}{n} (\hat{g}_a^{-1}(s)+\hat{g}_b^{-1}(s))  & \frac{1}{n}(\sqrt{\frac{n_b}{n_a}}\hat{g}_a^{-1}(s)-\sqrt{\frac{n_a}{n_b}}\hat{g}_b^{-1}(s))\\
    *& \hspace{-1.2cm}\frac{1}{n}(\frac{n_b}{n_a}\hat{g}_a^{-1}(s)+\frac{n_a}{n_b}\hat{g}_b^{-1}(s))+\lambda_2(L)f(s)\emt,
\end{align}
where
$$
    \hat{g}_a(s)=\lp\sum_{i}^{n_a}g_i^{-1}(s)\rp^{-1},\ \hat{g}_b(s)=\lp\sum_{i=n_a+1}^{n_a+n_b}g_i^{-1}(s)\rp^{-1}\,,
$$

are exactly the aggregate dynamics for each coherent group as defined in \eqref{eq_hat_g_def}. Such expression for $H_2(s)$ suggests that one may be able to represent $T_2(s)$ by some interconnection among $\hat{g}_a(s),\hat{g}_b(s)$ and $f(s)$, and this is indeed true:
\begin{figure*}[!t]
  \includegraphics[width=0.95\textwidth]{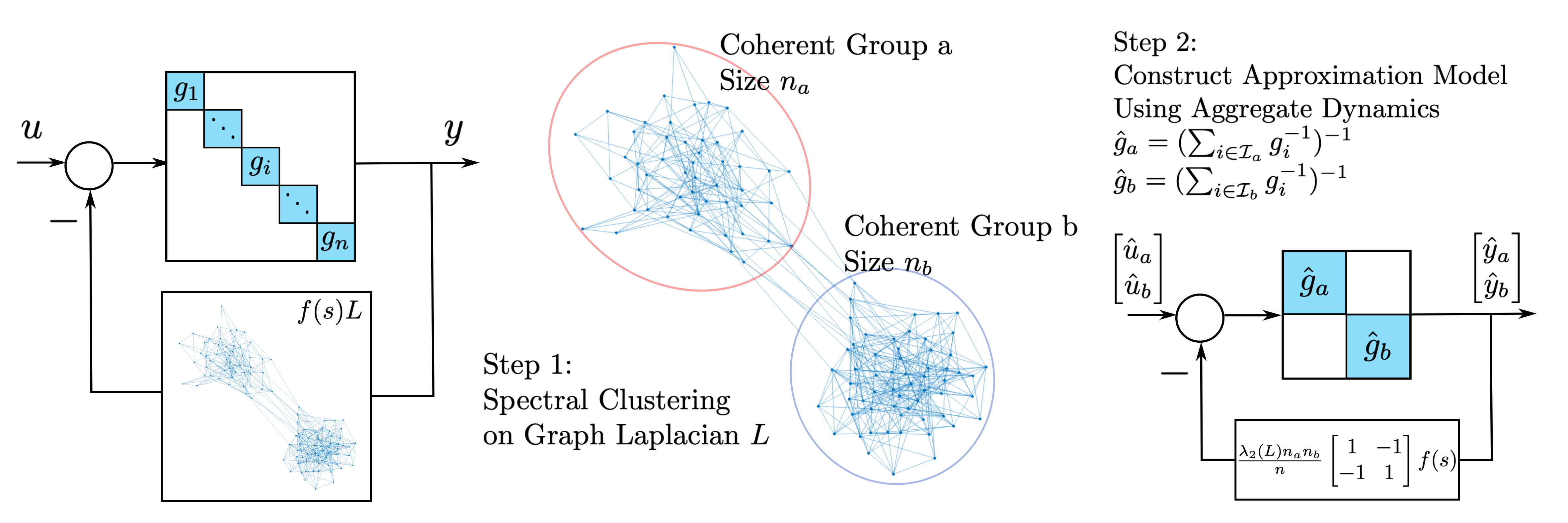}
  \caption{Illustration of our approximation model from Algorithm \ref{algo_approx_model}}
  \label{fig_algo_approx_model}
\end{figure*}
\begin{theorem}\label{thm_two_node_model}
    Consider a network of two nodes interconnected as in Fig. \ref{fig_blk_digm}, with node dynamics 
    $G(s):=\hat{G}(s)=\bmt \hat{g}_a(s) & 0\\
    0& \hat{g}_b(s)\emt$ and Laplacian matrix $L:=\hat{L}=\frac{\lambda_2(L_\text{blk})n_bn_a}{n}\bmt 1& -1\\ -1 &1 \emt$, and let $\hat{T}_2(s)$ be the $2\times 2$ transfer matrix from the network input $\hat{u}=\bmt \hat{u}_a \\ \hat{u}_b\emt$ to the output $\hat{y}=\bmt \hat{y}_a \\ \hat{y}_b\emt$. Then for networks with Laplacian matrix $L_\text{blk}$ defined in \eqref{eq_L_blk_def}, we have
    \be
        T_2(s)=\bmt \one_{n_a}& 0\\
        0 &\one_{n_b}\emt \hat{T}_2(s) \bmt \one_{n_a}^T & 0 \\ 0 &\one_{n_b}^T\emt\,.
    \ee
\end{theorem}
Here $\hat{T}_2(s)$ is precisely the dynamics of a network of two aggregate nodes $\hat{g}_a(s),\hat{g}_b(s)$ with the same network coupling dynamics $f(s)$ but with a new Laplacian matrix $\hat{L}$. Then $T_2(s)$ takes the aggregate inputs from each coherent group $\hat{u}=\bmt \sum_{i=1}^{n_a}u_i\\ \sum_{i=n_a+1}^{n_a+n_b}u_i\emt:=\bmt \hat{u}_a \\ \hat{u}_b\emt$ as the input to $\hat{T}_2(s)$, and its output $\bmt \hat{y}_a\one_{n_a}\\ \hat{y}_b\one_{n_b}\emt$ is coherent w.r.t. each group, where $\bmt \hat{y}_a\\ \hat{y}_b\emt=\hat{T}_2(s)\hat{u}$. Therefore, when $T_{yu}(s)$ can be well approximated by $T_2(s)$, the network dynamics can be understood by studying the interaction between two aggregate dynamics.

\section{Model Reduction for Networks under Weighted Stochastic Block Model}\label{sec_random_net_model}
We have shown that certain assumption on the graph Laplacian yields an interpretable reduced model, yet such a network is less practical as it requires dense connections among all the nodes. Can we ignore the fact that most practical networks do not have such dense connections and still build a two-node model from the same principle? If so, when do we expect such an approach to perform well?

Using Theorem \ref{thm_two_node_model}, we first propose our approximation model for networks with two coherent groups in Algorithm \ref{algo_approx_model}, where $\one_{\mathcal{I}}\in\mathbb{R}^n$ such that $[\one_\mathcal{I}]_i=\mathbf{1}_{\{i\in \mathcal{I}\}}=\begin{cases}
    1,& i \in\mathcal{I}\\
    0,& i\notin \mathcal{I}
\end{cases}$
for any $\mathcal{I}\subseteq [n]$. We also illustrate our algorithm in Figure \ref{fig_algo_approx_model}. The algorithm works for any network: it first finds tightly-connected components by Spectral Clustering on $L$, and then builds the two-node network as if $L$ has the same desired block structure as $L_{\text{blk}}$. Our analysis in Section \ref{sec_ideal_net_model} shows that such an algorithm will, for a network with $L_{\text{blk}}$, return the exact $T_2(s)$, which is in turn a good approximation for $T_{yu}(s)$ when $\lambda_3(L_{\text{blk}})$ is large. The question remains as to for what types of networks the algorithm performs well.
\begin{algorithm}[!h]
    \caption{Approximation Model $T_2(s)$ for Networks with Two Coherent Groups}
    \KwData{Network Model $\lp G(s)=\dg\{g_i(s)\}, L, f(s)\rp$}
    \textbf{Do}:
    \begin{enumerate}
        \item (Run Algorithm \ref{algo_sc})
        
        $(\mathcal{I}_a,\mathcal{I}_b,\lambda_2(L))\la \text{SpectralClustering}(L)$;
        \vspace{0.15cm}
        \item 
        $\hat{g}_a(s)\la \lp \sum_{i\in\mathcal{I}_a}g_i^{-1}(s)\rp^{-1},\ n_a\la |\mathcal{I}_a|$,
        
        $\hat{g}_b(s)\la \lp \sum_{i\in\mathcal{I}_b}g_i^{-1}(s)\rp^{-1},\ n_b\la |\mathcal{I}_b|$;

        \vspace{0.15cm}
        \item $\hat{G}(s)\la \bmt \hat{g}_a(s) &\hspace{-0.4cm} 0\\
        0& \hat{g}_b(s)\emt, \hat{L} \la\frac{\lambda_2(L)n_an_b}{n_a+n_b}\bmt 1 & \hspace{-0.4cm}-1\\
        -1 & 1\emt$,
        $\hat{T}_2(s)\la (I_2+f(s)\hat{G}(s)\hat{L})^{-1}\hat{G}(s)$;
        \vspace{0.15cm}
    \end{enumerate}
    \KwResult{$T_2(s)\la \bmt \one_{\mathcal{I}_a} & \one_{\mathcal{I}_b}\emt\hat{T}_2(s)\bmt \one^T_{\mathcal{I}_a}\\ \one^T_{\mathcal{I}_b}\emt$}
    \label{algo_approx_model}
\end{algorithm}
Recall that spectral clustering performs well on certain random graphs as long as the expected Adjacency matrix has the desired block structure~\cite{lyzinski2014perfect}. We argue that the same holds for our algorithm.

Now consider a random weighted graph of size $n$ whose adjacency matrix $A=[A_{ij}]$ is generated as
\begin{align}
    &\;A_{ij} =\begin{cases}
        W_{ij}, & \text{with probability } P_{ij} \\
        0, & \text{with probability } 1-P_{ij}
    \end{cases}\,, A_{ji}=A_{ij}\label{eq_A_gen}\\
    &\;\forall 1\leq i\leq j\leq n\,.\nonumber
\end{align}
Statistical graph theory often considers the case of unweighted graph, i.e., $W_{ij}=1,\forall i,j$, for which  \eqref{eq_A_gen} is called \emph{random graphs with independent edges} \cite{oliveira09}. Here, we require $A$ to be weighted to model the network coupling strength. One key result for such an independent edge model is that for large networks, the random adjacency matrix $A$ does not deviate from its expected value $\expc A$ too much, with high probability. Further, such concentration result can be extended for the Laplacian matrix $L_A$ as well. 
\begin{proposition}\label{prop_Lap_bd}
    Suppose $\max_{i,j}|W_{ij}|\leq 1$. Let $\Delta:=\max_{i}\sum_{j}P_{ij}W_{ij}^2$. For any $c>0$, If $\Delta\geq 16(c+1)\log n$, then for any $4n^{-c}\leq \delta <1$, we have
    $$
        \prob\lp \|L_{A}-L_{\expc A}\|\leq 8\sqrt{\Delta\log (4n/\delta)}\rp\geq 1-\delta\,.
    $$
\end{proposition}
We refer the readers to the Appendix for the proof. We make the following remarks. Firstly, this result is a generalization of~\cite[Theorem 3.1]{oliveira09}. Specifically,~\cite{oliveira09} considers the unweighted graph ($W_{ij}=1,\forall i,j$) and derived concentration results on the normalized Laplacian $D^{-\frac{1}{2}}LD^{-\frac{1}{2}}$, while our result works for weighted graph and we provide the concentration result regarding the original Laplacian $L$.
Secondly, the assumption $\max_{i,j}|W_{ij}|\leq 1$ is not critical as one can always scale $A$ by $\max_{i,j}|W_{ij}|$ and apply the result to the rescaled one. 
Last but not least, for the random graphs of our interests, we have $\Delta\sim_p \mathcal{O}(n)$,
then this Proposition essentially shows that with high probability, we have $\|L_A-L_{\expc A}\|\sim_p \mathcal{O}(\sqrt{n\log n})$, allowing us to relate the spectral properties of $L_A$ to those of $L_{\expc A}$.

Within this family of random graphs, we consider the one with two coherent groups:

\noindent
\emph{Weighted Stochastic Block Model with Two Communities} $(\mathcal{I}_a,\mathcal{I}_b,p,q,w_p,w_q)$: Given a partition $(\mathcal{I}_a,\mathcal{I}_b)$ of $[n]$, the adjacency matrix $A$ is generated as in \eqref{eq_A_gen} with
$$
    \begin{cases}
        P_{ij}=p, W_{ij}=w_p, & \mathcal{I}^{-1}(i)=\mathcal{I}^{-1}(j)\\
        P_{ij}=q, W_{ij}=w_q,& \mathcal{I}^{-1}(i)\neq\mathcal{I}^{-1}(j)
    \end{cases}\,,
$$
where $\mathcal{I}^{-1}(i)=\begin{cases}
    a, & i\in\mathcal{I}_a\\
    b, & i\in\mathcal{I}_b
\end{cases}$.

When $w_p=w_q=1$, this is exactly the stochastic block model with two communities~\cite{lyzinski2014perfect}. For the weighted version $w_p, w_q\geq 0$, notice that $\expc A=A_{\text{blk}}$ (up to a permutation matrix) with $\alpha = pw_p,\beta = qw_q, n_a=|\mathcal{I}_a|,n_b=|\mathcal{I}_b|$. Then $L_{\expc A}=L_{\text{blk}}$ is exactly the Laplacian matrix for the ideal network discussed in Section \ref{sec_ideal_net_model}. Given that $\|L_A-L_{\expc A}\|$ is small with high probability, we have the following result regarding the spectral properties of $L_A$:
\begin{theorem}\label{thm_spec_wsbm}
    Suppose $1\geq p> q\geq 0$, $w_q=\gamma w_p$ with $0<\gamma<1$. We let $n_{\min}:=\min\{|\mathcal{I}_a|,|\mathcal{I}_b|\}$ and $n=|\mathcal{I}_a|+|\mathcal{I}_b|$. If $n_{\min}=\Omega(n)$, then given any $0<\delta<1$ and  large enough $n$, with probability at least $1-\delta$ the following holds:
    \begin{enumerate}
        \item Large third smallest eigenvalue:
        \be
            \lambda_3(L_A)\geq w_p(p+\gamma q)n_{\min}-8w_p\sqrt{np\log(4n/\delta)}\label{eq_l3}
        \ee
        \item Approximately Good Invariant Subspace:
        \be
            \sin (v_2(L_A),v_2(L_{\expc A}))\leq \frac{16\sqrt{2}}{p-\gamma q} \sqrt{\frac{np\log(4n/\delta)}{n^2_{\min}}}\,,\label{eq_v2}
        \ee
        where $$\sin (v_2(L_A),v_2(L_{\expc A}))=\sqrt{1-|v_2(L_A)^Tv_2(L_{\expc A})|^2}\,.$$
    \end{enumerate}
\end{theorem}
\begin{proof}[Proof sketch]
    With the upper bound $\|L_{A}-L_{\expc A}\|\sim \mathcal{O}_p(\sqrt{n\log n})$ from Proposition \ref{prop_Lap_bd}. The first result is due to Weyl's inequality~\cite{Horn:2012:MA:2422911}, that is $|\lambda_3(L_A)-\lambda_3(L_{\expc A})|\leq \|L_{A}-L_{\expc A}\|$.  The second results is a direct result of a variant of Davis-Khan Theorem in~\cite{Yu2014}, that is $\sin (v_2(L_A),v_2(L_{\expc A}))\sim\mathcal{O}_p\lp \frac{\|L_{A}-L_{\expc A}\|}{\lambda_3(L_{\expc A})-\lambda_2(L_{\expc A})}\rp$. We refer the readers to the Appendix for the proof.
\end{proof}
For large networks, if the sizes of two coherent groups are balanced so that $n_{\min}\sim \Omega_p(n)$. Then \eqref{eq_l3} implies $\lambda_3(L_A)\sim \Omega_p (n)$, showing that $T_2(s)$ will be a good approximation to $T_{yu}(s)$, by Theorem \ref{thm_T2}. Also, \eqref{eq_v2} shows that $\sin (v_2(L_A),v_2(L_{\expc A}))\sim\mathcal{O}_p\lp \sqrt{\frac{\log n}{n}}\rp$, hence the $T_2(s)$ constructed from $v_2(L_A)$ is close to the one constructed from $v_2(L_{\expc A})$, which is exactly the two-node model from Theorem \ref{thm_two_node_model}. Therefore, we should expect Algorithm \ref{algo_approx_model} to perform well for the weighted stochastic block model, even if one instance of such random graph appears much different than an ideal $A_{\text{blk}}$ as it has much fewer edges.

\section{Application: Modeling Frequency Response in Power Networks}\label{sec_num}
The frequency response of synchronous generator (including grid-forming inverters) networks, linearized at its equilibrium point~\cite{zhao2013power}, can be modeled exactly as the network model in Fig \ref{fig_blk_digm} with $f(s)=\frac{1}{s}$ and
\begin{align*}
    u_i&\text{: Disturbance in mechanical power at generator $i$}\\
    \vspace{-0.2cm}
    y_i&\text{: Frequency of generator $i$}\\
    \vspace{-0.2cm}
    g_i(s)&\text{: Generator dynamics}\\
    \vspace{-0.2cm}
    L &\text{: Sensitivity of power injection w.r.t. bus phase angles}
\end{align*}
As for generator dynamics, we use the second-order model
\be g_i(s)=\frac{1}{m_is+d_i+\frac{r_i^{-1}}{\tau_i s+1}}\,,\label{eq_gen_dym_1st}\ee
where $m_i$ is the inertia, $d_i$ the damping, $r_i^{-1}$ the droop coefficient,  and $\tau_i$ the turbine time constant of generator $i$.
\begin{rem}
Aggregating generators with second-order dynamics do not returns a $\hat{g}_a(s)$ with the same order as a single generator if the turbine time constant $\tau_i$ are different across generators, then one may need to utilize model reduction techniques such as balanced truncation~\cite{Zhou:1996:ROC:225507} on $\hat{g}_a(s)$, please refer to~\cite{min2020lcss} for a detailed discussion. In the experiment, we do not do model reduction on $\hat{g}_a(s)$.
\end{rem}
\subsection{Synthetic Case: Weighted Stochastic Block Model}
We first validate our algorithm with a synthetic test case, where generator dynamics follows \eqref{eq_gen_dym_1st} and we randomly sample the inertia and damping independently as
\begin{align*}
    &\;m_i\sim \text{Uniform}([0.05,0.5]),  d_i\sim \text{Uniform}([0.2,0.5])\,.\\
    &\;r_i\sim \text{Uniform}([5,10]), \tau_i\sim\text{Uniform}([2,10])
\end{align*}
The adjacency matrix $A$ is sampled from our weighted stochastic block model $(\mathcal{I}_a,\mathcal{I}_b,p,q,w_p,w_q)$ where $n_a=|\mathcal{I}_a|=30,n_b=|\mathcal{I}_b|=20,p=0.6,q=0.1,w_p=5,w_q=0.5$. We note that for spectral clustering, Algorithm \ref{algo_sc} always achieves perfect clustering across multiple runs. With the generated network model, we inject a step disturbance $u_2(t)=\chi(t)$ at the second node and plot the step response of $T_{yu}(s)$ in Fig \ref{fig_syn}, along with the response $\hat{y}_a,\hat{y}_b$ of our approximate model $T_2(s)$ from Algorithm \ref{algo_approx_model}. There is a clear difference between the dynamical response of generators from group $a$ and group $b$, and the aggregate responses $\hat{y}_a,\hat{y}_b$ capture such difference while providing a good approximation to the actual node responses.
Due to space constraints, we only present the result of running Algorithm \ref{algo_approx_model} on one instance of the randomly generated networks, but the results are consistent across multiple runs. 
\begin{figure}[!h]
    \centering
    \includegraphics[width=\columnwidth]{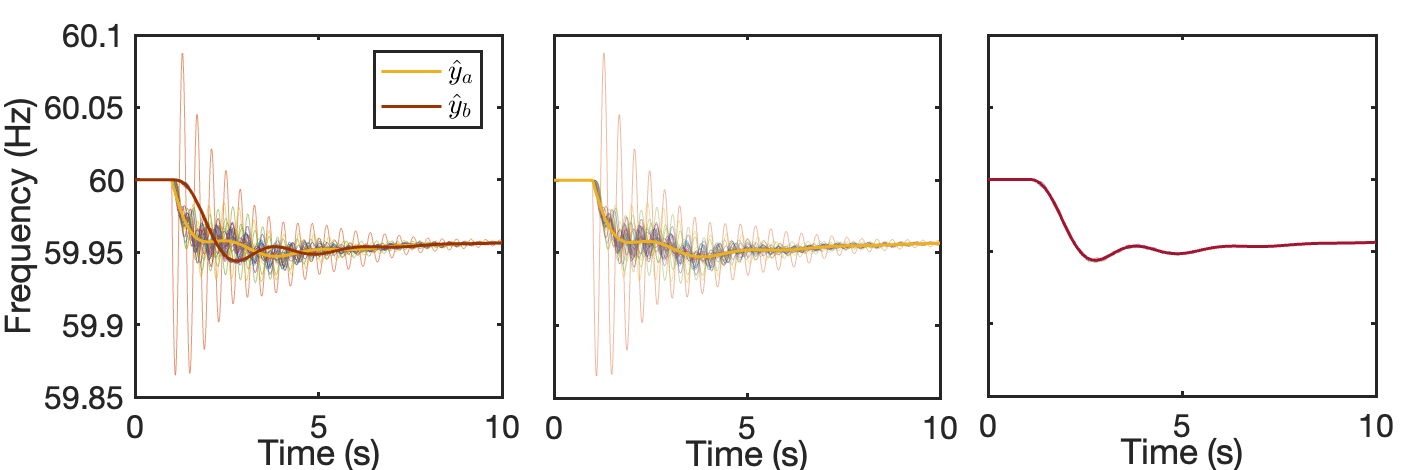}
    \caption{Synthetic Dataset ($p=0.6$): Step Response of $T_{yu}(s)$ and $T_2(s)$ from algorithm \ref{algo_approx_model}. (Middle) Response shown for only group $a$, (Right) Response shown for only group $b$. The node injected with step disturbance is in the group $a$.}
    \label{fig_syn}
\end{figure}

Our theorem suggests that if we increase $p$, i.e., having more intra-group connection, then $\lambda_3(L_A)$ increases, which makes our approximation model closer to the true network. Indeed, we run the same experiment with $p=0.9$, we see a more coherent behavior in the network, as shown in Figure \ref{fig_syn_tight}.
\begin{figure}[!h]
    \centering
    \includegraphics[width=\columnwidth]{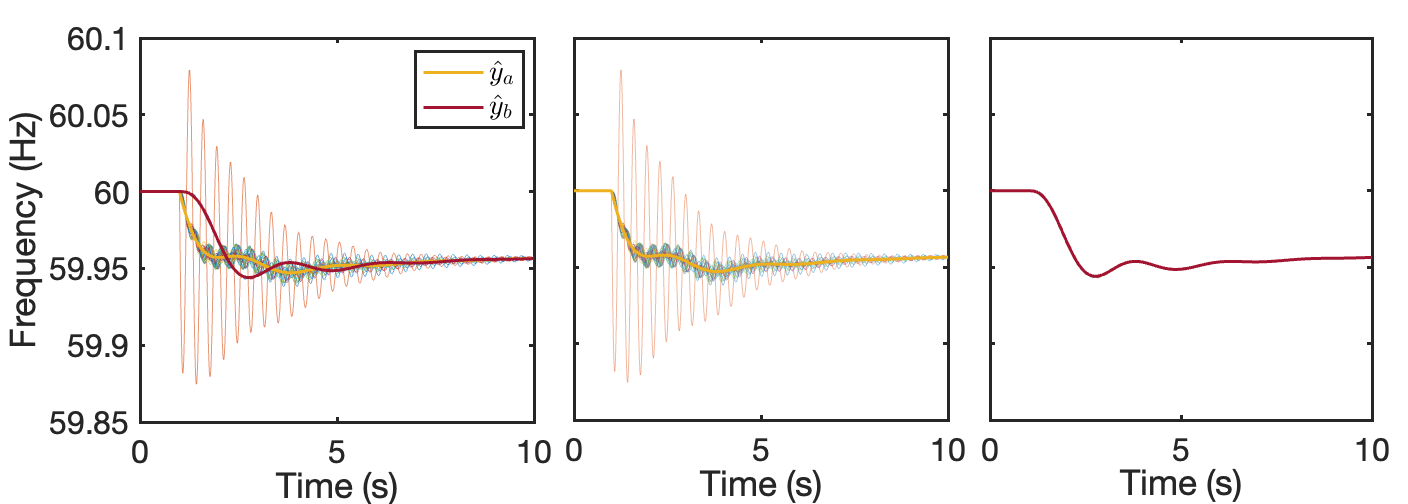}
    \caption{Synthetic Dataset ($p=0.9$): Step Response of $T_{yu}(s)$ and $T_2(s)$ from algorithm \ref{algo_approx_model}. (Middle) Response shown for only group $a$, (Right) Response shown for only group $b$. The node injected with step disturbance is in the group $a$.}
    \label{fig_syn_tight}
\end{figure}
\subsection{Test Case: IEEE 68-bus System}
Lastly, we apply our algorithm to the IEEE 68-bus test system~\cite{pal2006robust}. We first use Kron reduction to eliminate the load buses in the network~\cite{dorfler2013} and apply our algorithm to the reduced network with only the generator buses. The spectral clustering result (from Algorithm \ref{algo_sc}) correctly identifies the two areas in the test system, one for NETS and one for NYPS and its adjacent areas.
\begin{figure}[t]
    \centering
    \includegraphics[width=0.9\columnwidth]{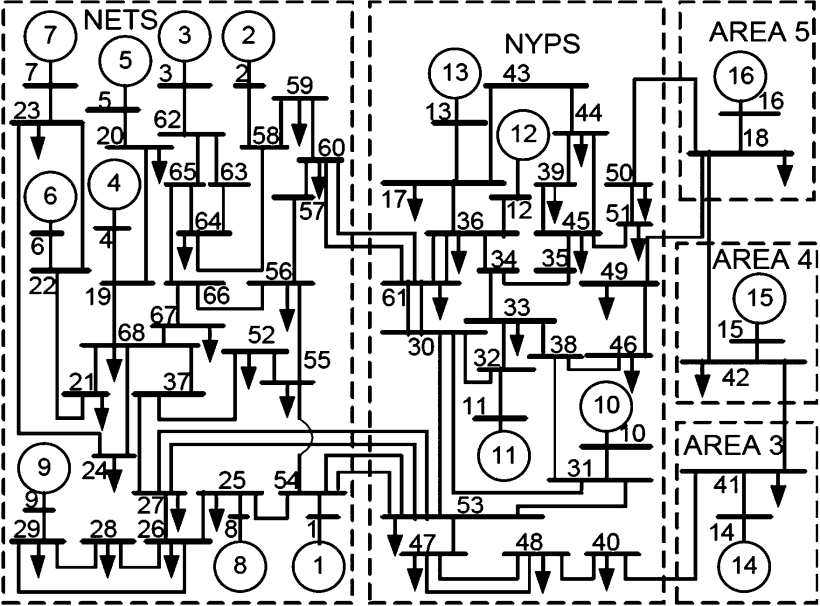}
    \caption{IEEE 68-bus test system~\cite{pal2006robust}. Image credits to~\cite{khalid16}. (NETS: New England test system, NYPS: New York Power System)}
    \label{ieee_68}
\end{figure}


We also inject a step disturbance $\chi(t)$ to the second generator in the network and plot the response of $T_{yu}(s)$ and $T_2(s)$. While the responses appear less coherent than the synthetic case, our approximation still captures the average trend in each group. Notably, the approximation model captures the underdamped response of coherent group $a$, when the disturbance is injected to group $b$.
\begin{figure}[!h]
    \centering
    \includegraphics[width=\columnwidth]{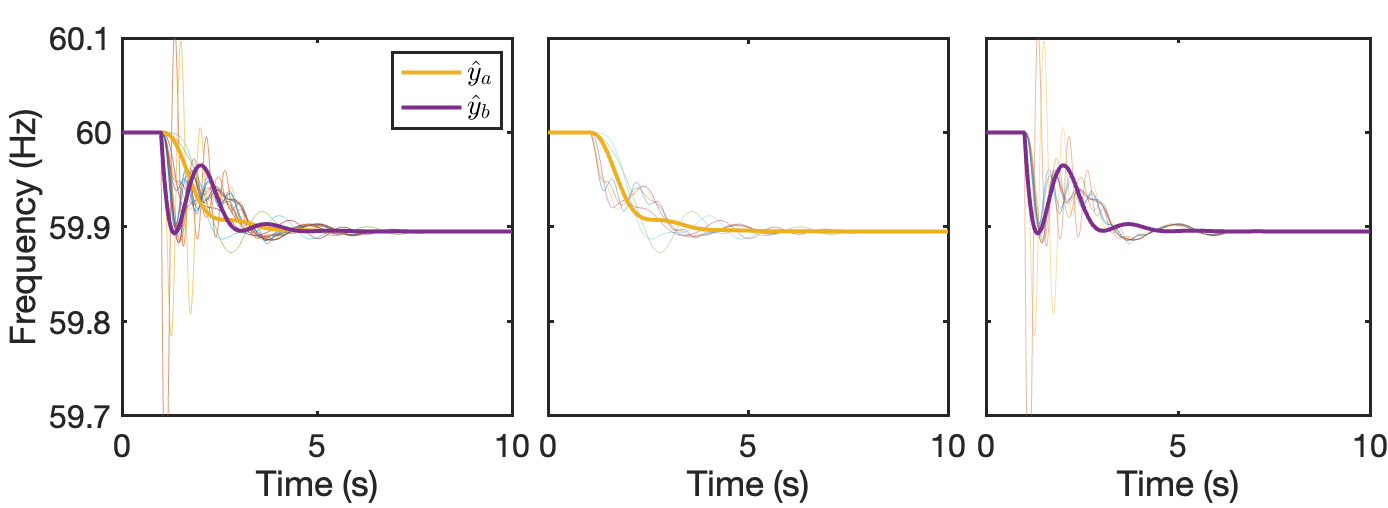}
    \caption{IEEE 68-bus system (16 generators): Step Response of $T_{yu}(s)$ and $T_2(s)$ from algorithm \ref{algo_approx_model}. (Middle) Response shown for only group $a$, (Right) Response shown for only group $b$. The generator injected with step disturbance is in group $b$}
    \label{fig_ieee_68}
\end{figure}
\section{Conclusion}
In this paper, we propose a novel model-reduction methodology for large-scale dynamic networks, based on the recent frequency-domain characterization of coherent dynamics in networked systems. Our analysis shows that networks with two coherent groups/areas can be well approximated by two aggregate nodes dynamically interacting with each other, and our results are theoretically justified for networks with ideal block structures and networks with random graphs generated from a weighted stochastic block model. The numerical results align with our theoretical findings. We believe our analysis can be extended to networks with more than two coherent groups and designing algorithms to identify multiple coherent groups and modeling their interaction are interesting future research directions.
\bibliographystyle{IEEEtran}
\bibliography{ref.bib}
\newpage
\onecolumn
\appendix
\input{appendix.tex}
\end{document}

%% file: edits.tex
\usepackage{ifthen}
\newboolean{showcomments}
\setboolean{showcomments}{true}
\usepackage[usenames,dvipsnames]{xcolor}
\usepackage{todonotes}

\definecolor{bleudefrance}{rgb}{0.19, 0.55, 0.91}
\definecolor{ao(english)}{rgb}{0.0, 0.5, 0.0}

\newcommand{\addcite}[0]{\ifthenelse{\boolean{showcomments}}
{\textcolor{purple}{(add cite(s)) }}{}}%

\newcommand{\hl}[1]{\ifthenelse{\boolean{showcomments}}
{\textcolor{red}{#1}}{#1}}
\newcommand{\enrique}[1]{  \ifthenelse{\boolean{showcomments}}
{\todo[inline,color=bleudefrance]{Enrique: #1}}{}}
\newcommand{\emmargin}[1]{\ifthenelse{\boolean{showcomments}}{\marginpar{\color{bleudefrance}\tiny EM: #1}}{}}
\newcommand{\hancheng}[1]{  \ifthenelse{\boolean{showcomments}}
{\todo[inline,color=orange]{Hancheng: #1}}{}}

\newboolean{showedits}
\setboolean{showedits}{true}

\usepackage[xcolor=bleudefrance]{changes}
\definechangesauthor[color=bleudefrance]{EM}
\newcommand{\aem}[1]{
\ifthenelse{\boolean{showedits}}
{\added[id=EM]{#1}}
{#1}
}
\newcommand{\chem}[2]{
\ifthenelse{\boolean{showedits}}
{\replaced[id=EM]{#1}{#2}}
{#1}
}
\newcommand{\dem}[1]{
\ifthenelse{\boolean{showedits}}
{\deleted[id=EM]{#1}}
{}
}

%% file: appendix.tex
\subsection{Proof of Theorem \ref{thm_T2}}
\begin{proof}[Proof of Theorem \ref{thm_T2}]
    Firstly, we have
    \begin{align*}
        T_{yu}(s_0)&=\;(I_n+G(s_0)f(s_0)L)^{-1}G(s_0)\\
        &=\; (G^{-1}(s_0)+f(s_0)L)^{-1}\\
        &=\; V(V^TG^{-1}(s_0)V+f(s_0)\Lambda )^{-1}V^T\,,
    \end{align*}
    where $G^{-1}(s_0)=\dg\{g^{-1}_i(s_0)\}$, $\Lambda=\dg\{\lambda_i(L)\}$, and $V=\bmt v_1(L),v_2(L),\cdots,v_n(L)\emt$.
    
    Let $H=V^T\dg\{g_i^{-1}(s_0)\}V+f(s_0)\Lambda$, then 
    \ben
    T_{yu}(s_0) = VH^{-1}V^T.
    \een
    Then it is easy to see that
    \begin{align}
        \lV T_{yu}(s_0)-T_2(s_0)\rV&=\; \lV T(s_0)-V\bmt H_2^{-1} &0 \\
        0& 0\emt V^T\rV\nonumber\\
        &=\; \lV V\lp H^{-1}-\bmt H_2^{-1} &0 \\
        0& 0\emt\rp V^T\rV\nonumber\\
        &=\; \lV H^{-1}-\bmt H_2^{-1} &0 \\
        0& 0\emt \rV\,,\label{eq_T_H_norm_equiv}
    \end{align}
    where the last equality comes from the fact that multiplying by a unitary matrix $V$ preserves the spectral norm.
    
    Let $V_2:=\bmt\frac{\one}{\sqrt{n}} & v_2(L)\emt$ and $ V_2^\perp:=\bmt v_3(L)& \cdots &v_n(L)\emt$, we now write $H$ in block matrix form:
    \begin{align*}
        H&=\;V^T\dg\{g_i^{-1}(s_0)\}V+f(s_0)\Lambda\\
        &= \bmt
            V_2^T\\
            (V_2^\perp)^T
        \emt \dg\{g_i^{-1}(s_0)\} \bmt
        V_2& V_2^\perp\emt+f(s_0)\Lambda\\
        &= {\small\bmt
        H_2& V_2^T\dg\{g_i^{-1}(s_0)\}V_2^\perp\\
        (V_2^\perp)^T\dg\{g_i^{-1}(s_0)\}V_2 &(V_2^\perp)^T\dg\{g_i^{-1}(s_0)\}V_2^\perp+f(s_0)\Tilde{\Lambda}
        \emt}\\
        &:= \bmt
        H_2& H^T_o\\
        H_o & H_d
        \emt\,,
    \end{align*}
    where  $\Tilde{\Lambda}=\dg\{\lambda_3(L),\cdots,\lambda_n(L)\}$.
    
    Inverting $H$ in its block form, we have
    \ben
        H^{-1} = \bmt
        A &-A H_o^TH_d^{-1}\\
        -H_d^{-1}H_o A& H_d^{-1}+H_d^{-1}H_o A H_o^TH_d^{-1}
        \emt\,,
    \een
    where $a = (H_2-H_o^TH_d^{-1}H_o)^{-1}$.
    
    Notice that $||V_2^\perp||=1$ and $||V_2||=1$, we have
    \begin{align}
        \|H_o\|&=\; \lV (V_2\perp)^T\dg\{g_i^{-1}(s_0)\}V_2\rV\nonumber\\
        &\leq\; \|V_\perp\|\|\dg\{g_i^{-1}(s_0)\}\|\|V_2\|\leq M_2\,.\label{eq_h12_norm_bd}
    \end{align}

    Also, by Weyl's inequality~\cite{Horn:2012:MA:2422911}, when $|f(s_0)|\lambda_2(L)>M_2$, the following holds:
    \begin{align}
        \|H_d^{-1}\|&=\;\|(f(s_0)\Tilde{\Lambda}+ (V_2^\perp)^T\dg\{g_i^{-1}(s_0)\}V_2^\perp)^{-1}\|\nonumber\\
        &\leq\; \frac{1}{\sigma_1(f(s_0)\Tilde{\Lambda})-\|(V_2^\perp)^T\dg\{g_i^{-1}(s_0)\}V_2^\perp\|}\nonumber\\
        &\leq\; \frac{1}{\sigma_1(f(s_0)\Tilde{\Lambda})-M_2}\leq \frac{1}{|f(s_0)|\lambda_3(L)-M_2} \,.\label{eq_H22_norm_bd}
    \end{align}
    
    Lastly, when $|f(s_0)|\lambda_2(L)>M_2+M_2^2M_1$, a similar reasoning as above, using \eqref{eq_h12_norm_bd} \eqref{eq_H22_norm_bd}, and our assumption $\|T_2(s_0)\|=\|H_2^{-1}\|\leq M_1$, gives
    \begin{align}
        \|A\|&\leq\; \frac{1}{\|H_2\|-\|H_o^TH_d^{-1}H_o\|}\nonumber\\
        &\leq\; \frac{1}{\|H_2\|-\|H_o\|^2\|H_d^{-1}\|}\nonumber\\
        &\leq\; \frac{1}{\frac{1}{M_1}-\frac{M_2^2}{|f(s_0)|\lambda_3(L)-M_2}}\nonumber\\
        &= \;
        \frac{(|f(s_0)|\lambda_3(L)-M_2)M_1}{|f(s_0)|\lambda_3(L)-M_2-M_1M_2^2}\,.\label{eq_a_norm_bd}
    \end{align}
    
    Now we bound the norm of $H^{-1}-\bmt H_2^{-1} &0 \\
        0& 0\emt$ by the sum of norms of all its blocks:
    \begin{align}
        &\;\lV H^{-1}-\bmt H_2^{-1} &0 \\
        0& 0\emt \rV\nonumber\\
        =&\; \lV \bmt
        AH_o^TH_d^{-1}H_o H_2 &-aH_o^TH_d^{-1}\\
        -aH_d^{-1}H_o& H_d^{-1}+aH_d^{-1}H_oH_o^TH_d^{-1}
        \emt\rV\nonumber\\
        \leq &\; \|AH_o^TH_d^{-1}H_o H_2\|+2\|AH_d^{-1}H_o\|\nonumber\\
        &\; +\|H_d^{-1}+AH_d^{-1}H_oH_o^TH_d^{-1}\|\nonumber\\
        \leq &\; \|A\|\|H_d^{-1}\|(\|H_2\|\|H_o\|^2+2\|H_o\|+\|H_o\|^2\|H_d^{-1}\|)\nonumber\\
        &\;\quad\quad +\|H_d^{-1}\|\,,\label{eq_Hinv_norm_bd1}
    \end{align}
    Using \eqref{eq_h12_norm_bd}\eqref{eq_H22_norm_bd}\eqref{eq_a_norm_bd}, we can further upper bound \eqref{eq_Hinv_norm_bd1} as
    \begin{align}
        &\;\lV H^{-1}-\bmt H_2^{-1} &0 \\
        0& 0\emt \rV\nonumber\\
        \leq &\;  \frac{M_1^2M_2^2+2M_1M_2+\frac{M_1M_2^2}{|f(s_0)|\lambda_3(L)-M_2}}{|f(s_0)|\lambda_3(L)-M_2-M_1M_2^2}+\frac{1}{|f(s_0)|\lambda_2(L)-M_2}\nonumber\\
        =&\;\frac{\lp M_1M_2+1\rp^2}{|f(s_0)|\lambda_3(L)-M_2-M_1M_2^2}\,.\label{eq_Hinv_norm_bd2}
    \end{align}
    This bound holds as long as $|f(s_0)|\lambda_3(L)>M_2+M_2^2M_1$. Combining \eqref{eq_T_H_norm_equiv} and \eqref{eq_Hinv_norm_bd2} gives the desired inequality.
\end{proof}
\newpage
\subsection{Eigenvalues and Eigenvectors of $L_{\text{blk}}(n_a,n_b,\alpha,\beta)$}\label{app_eig_lap}
The matrix $A_{\text{blk}}(n_a,n_b,\alpha,\beta)$ is defined to be 
$$
    A_{\text{blk}}(n_a,n_b,\alpha,\beta)=\bmt\alpha \one_{n_a}\one_{n_a}^T & \beta \one_{n_a}\one_{n_b}^T\\
    \beta \one_{n_b}\one_{n_a}^T&\alpha \one_{n_b}\one_{n_b}^T\emt\,,
$$
and $L_{\text{blk}}(n_a,n_b,\alpha,\beta)=\dg\{A_{\text{blk}}\one_{n_a+n_b}\}-A_{\text{blk}}$.

Consider any non-zero vector $v$ such that $L_{\text{blk}} v =\lambda v$ for some $\lambda\geq 0$. We write $v=\bmt v_{1:n_a}\\ v_{n_a+1:n_b}\emt:=\bmt v_1\\ v_2\emt$. Then $L_{\text{blk}} v =\lambda v$ can be written as
\be
    \bmt (n_a\alpha+n_b\beta)v_1 \\ (n_b\alpha+n_a\beta) v_2\emt-\bmt \lp\alpha (\one_{n_a}^Tv_1)+\beta(\one_{n_b}^Tv_2)\rp\one_{n_a}\\ \lp\alpha (\one_{n_b}^Tv_2)+\beta(\one_{n_a}^Tv_1)\rp\one_{n_b}\emt=\lambda \bmt v_1\\ 
     v_2\emt\,.\label{eq_L_eig_eq}
\ee
Multiply $\bmt \one_{n_a}^T & 0\\ 0 &\one_{n_b}^T\emt$ to the left of \eqref{eq_L_eig_eq}, we have
$$
    \bmt (n_a\alpha+n_b\beta)(\one_{n_a}^Tv_1) \\ (n_b\alpha+n_a\beta) (\one_{n_b}^Tv_2)\emt-\bmt \lp n_a\alpha (\one_{n_a}^Tv_1)+n_a\beta(\one_{n_b}^Tv_2)\rp\\ \lp n_b\alpha (\one_{n_b}^Tv_2)+n_b\beta(\one_{n_a}^Tv_1)\rp\emt=\lambda\bmt  (\one_{n_a}^Tv_1)\\ 
    (\one_{n_b}^Tv_2)\emt\,,
$$
which leads to
$$
    \bmt n_b\beta-\lambda & -n_a\beta \\
    -n_b\beta & n_a\beta -\lambda\emt \bmt (\one_{n_a}^Tv_1)\\ 
    (\one_{n_b}^Tv_2)\emt=0\,.
$$
We view the equation above as a system of linear equations:

\noindent
\textbf{When it has non-zero solution $\bmt (\one_{n_a}^Tv_1)\\ 
    (\one_{n_b}^Tv_2)\emt\neq 0$}: It implies $\det\lp \bmt n_b\beta-\lambda & -n_a\beta \\
    -n_b\beta & n_a\beta -\lambda\emt\rp=0$, which is $\lambda(\lambda-(n_a+n_b)\beta)=0$. Therefore $\lambda=0$ or $\lambda=(n_a+n_b)\beta$, this gives two eigenpair:
    $$
        \lp\lambda=0,v=\frac{1}{\sqrt{n}}\one_{n}\rp,\ \text{or }\lp\lambda=(n_a+n_b)\beta,v=\frac{1}{\sqrt{n}}\bmt \sqrt{\frac{n_b}{n_a}}\one_{n_a}\\ -\sqrt{\frac{n_a}{n_b}}\one_{n_b}\emt\rp\,,
    $$
    where $n=n_a+n_b$, and eigenvector $v$ is normalized, i.e., $\|v\|=1$.

\noindent
\textbf{When it has only zero solution $\bmt (\one_{n_a}^Tv_1)\\ 
    (\one_{n_b}^Tv_2)\emt= 0$}: When both $(\one_{n_a}^Tv_1), (\one_{n_b}^Tv_2)$ are zero, \eqref{eq_L_eig_eq} reduces to
    $$
        \bmt (n_a\alpha+n_b\beta)v_1 \\ (n_b\alpha+n_a\beta) v_2\emt=\lambda \bmt v_1\\ 
     v_2\emt\,.
    $$
    
    \underline{When either $n_a=n_b$ or $\alpha=\beta$}, one have $n_a\alpha+n_b\beta = n_b\alpha+n_a\beta =\lambda$. This is an simple eigenvalue with algebraic multiplicity $n_a+n_b-2$.
    
    \underline{When $n_a\neq n_b$ and $\alpha\neq\beta$}, we have $n_a\alpha+n_b\beta \neq  n_b\alpha+n_a\beta$. In this case, $v_1$ and $v_2$ can not be non-zero at the same time and 
    \begin{enumerate}
        \item $\lp\lambda=n_a\alpha+n_b\beta, v=\bmt v_1\\ 0\emt\rp$ is an eigenpair for any $v_1$ such that $\one_{n_a}^Tv_1=0,\|v_1\|=1$. 
        \item $\lp\lambda=n_b\alpha+n_a\beta, v=\bmt 0\\ v_2\emt\rp$ is an eigenpair for any $v_2$ such that $\one_{n_b}^Tv_2=0,\|v_2\|=1$.
    \end{enumerate}
    Then $\lambda=n_a\alpha+n_b\beta$ is an simple eigenvalue with algebraic multiplicity $n_a-1$ and so is $\lambda=n_b\alpha+n_a\beta$ with algebraic multiplicity $n_b-1$. 
    
    So far we find all eigenvalues and eigenvectors of $L_{\text{blk}}$.
    
    Notice that when $\alpha\geq \beta$, we have $\min\{n_a\alpha+n_b\beta, n_b\alpha+n_a\beta\}\geq (n_a+n_b)\beta$. In this case, the first two smallest eigenvalues and their corresponding eigenvectors are
    $$
        \lp\lambda_1(L_{\text{blk}})=0,v_1(L_{\text{blk}})=\frac{1}{\sqrt{n}}\one_{n}\rp,\ \text{or }\lp\lambda_2(L_{\text{blk}})=(n_a+n_b)\beta,v_2(L_{\text{blk}})=\frac{1}{\sqrt{n}}\bmt \sqrt{\frac{n_b}{n_a}}\one_{n_a}\\ -\sqrt{\frac{n_a}{n_b}}\one_{n_b}\emt\rp\,.
    $$
    and $\lambda_3(L_{\text{blk}})=\min\{n_a\alpha+n_b\beta, n_b\alpha+n_a\beta\}$.
\newpage
\subsection{Proofs of Theorem \ref{thm_two_node_model}}
\begin{proof}[Proof of Theorem \ref{thm_two_node_model}]
    Since $v_2(L_\text{blk})=\frac{1}{\sqrt{n}}\bmt \sqrt{\frac{n_b}{n_a}}\one_{n_a}\\ -\sqrt{\frac{n_a}{n_b}}\one_{n_b}\emt$, we have
    $$
        T_2(s)=\frac{1}{n}\bmt \one_{n_a} & \sqrt{\frac{n_b}{n_a}}\one_{n_a}\\ \one_{n_b} & -\sqrt{\frac{n_a}{n_b}}\one_{n_b}\emt H_2^{-1}\bmt \one_{n_a} & \sqrt{\frac{n_b}{n_a}}\one_{n_a}\\ \one_{n_b} & -\sqrt{\frac{n_a}{n_b}}\one_{n_b}\emt^T\,,
    $$
    where
    $$
        H_2(s)=\bmt \frac{1}{n}(\hat{g}_a^{-1}(s)+\hat{g}_b^{-1}(s)) & \frac{1}{n}(\sqrt{\frac{n_b}{n_a}}\hat{g}_a^{-1}(s)-\sqrt{\frac{n_a}{n_b}}\hat{g}_b^{-1}(s))\\
    \frac{1}{n}(\sqrt{\frac{n_b}{n_a}}\hat{g}_a^{-1}(s)-\sqrt{\frac{n_a}{n_b}}\hat{g}_b^{-1}(s)) & \frac{1}{n}(\frac{n_b}{n_a}\hat{g}_a^{-1}(s)+\frac{n_a}{n_b}\hat{g}_b^{-1}(s))+\lambda_2(L_\text{blk})f(s)\emt,
    $$
    Notice that
    $$
        \bmt \one_{n_a} & \sqrt{\frac{n_b}{n_a}}\one_{n_a}\\ \one_{n_b} & -\sqrt{\frac{n_a}{n_b}}\one_{n_b}\emt \bmt \frac{n_a}{n} & \frac{n_b}{n}  \\
        \sqrt{\frac{n_a n_b}{n^2}} & -\sqrt{\frac{n_a n_b}{n^2}} \emt=\bmt \one_{n_a} & 0\\
        0 & \one_{n_b}\emt\,.
    $$
    Then
    $$
        T_2(s)=\bmt \one_{n_a} & 0\\
        0 & \one_{n_b}\emt \frac{1}{n}\bmt \frac{n_a}{n} & \frac{n_b}{n}  \\
        \sqrt{\frac{n_a n_b}{n^2}} & -\sqrt{\frac{n_a n_b}{n^2}} \emt^{-1}H_2^{-1}(s)\bmt \frac{n_a}{n} & \frac{n_b}{n}  \\
        \sqrt{\frac{n_a n_b}{n^2}} & -\sqrt{\frac{n_a n_b}{n^2}} \emt^{-T}\bmt \one_{n_a} & 0\\
        0 & \one_{n_b}\emt^T\,.
    $$
    Now
    \begin{align*}
        &\;\frac{1}{n}\bmt \frac{n_a}{n} & \frac{n_b}{n}  \\
        \sqrt{\frac{n_a n_b}{n^2}} & -\sqrt{\frac{n_a n_b}{n^2}} \emt^{-1}H_2^{-1}(s)\bmt \frac{n_a}{n} & \frac{n_b}{n}  \\
        \sqrt{\frac{n_a n_b}{n^2}} & -\sqrt{\frac{n_a n_b}{n^2}} \emt^{-T}\\
        =&\; \lp n \bmt \frac{n_a}{n} & \frac{n_b}{n}  \\
        \sqrt{\frac{n_a n_b}{n^2}} & -\sqrt{\frac{n_a n_b}{n^2}} \emt^TH_2(s) \bmt \frac{n_a}{n} & \frac{n_b}{n}  \\
        \sqrt{\frac{n_a n_b}{n^2}} & -\sqrt{\frac{n_a n_b}{n^2}} \emt\rp^{-1}\\
        =&\;\lp\bmt \hat{g}_a^{-1}(s) & 0 \\
        0 & \hat{g}_b^{-1}(s)\emt+\lambda_2(L_\text{blk})\frac{n_an_b}{n}\bmt 1 & -1 \\ -1 &1\emt f(s)\rp^{-1}\\
        =&\; (\hat{G}^{-1}(s)+\hat{L}f(s))^{-1}\\
        =&\; (I_2+\hat{G}(s)\hat{L}f(s))^{-1}\hat{G}(s):=\hat{T}_2(s)\,.
    \end{align*}
    Therefore $T_2(s)=\bmt \one_{n_a} & 0\\
        0 & \one_{n_b}\emt \hat{T}_2(s)\bmt \one_{n_a} & 0\\
        0 & \one_{n_b}\emt^T$, where $\hat{T}_2(s)$ is exactly a network model with two nodes.
\end{proof}
\newpage
\subsection{Proofs for Section \ref{sec_random_net_model}}
\begin{lemma}[Corollary 7.1 in~\cite{oliveira09}]\label{lem_mat_concentration}
    Let $X_1,\cdots,X_n$ be mean-zero independent random $d\times d$ Hermitian matrices and such that there exists a $M>0$ with $\|X_i\|\leq M$ almost surely for $1\leq i\leq m$. Define $\sigma^2\equiv \lambda_{\max}\lp \sum_{i=1}^n\expc [X_i^2]\rp$. Then for all $t>0$,
    $$
    \prob\lp \lV \sum_{i=1}^nX_i\rV\geq t\rp\leq 2d\exp\lp-\frac{t^2}{8\sigma^2+4Mt}\rp\,.
    $$
\end{lemma}
\begin{lemma}[Direct consequence of Theorem 3.3 in~\cite{chung2006concentration}]\label{lem_weighted_chernoff}
    Let $X_1,\cdots,X_n$ be independent Bernoulli random variables with $\prob\lp X_i=1\rp=p_i, 1\leq i\leq n$. For $X=\sum_{i=1}^na_iX_i$ with $a_i>0$, we define $\nu=\sum_{i=1}^na_i^2p_i$. Then, we have
    $$
        \prob\lp |X-\expc[X]|\geq t\rp\leq \exp\lp -\frac{t^2}{2(\nu+a_{\max} t/3)}\rp\,,
    $$
    where $a_{\max}=\max_i a_i$.
\end{lemma}

\begin{proof}[Proof of Proposition \ref{prop_Lap_bd}]
    We let $W=[W_{ij}],P=[P_{ij}]$, where $W_{ji}=W_{ij},P_{ji}=P_{ij}, \forall j>i$. 
    Notice that $\expc A=W\odot P$.
    
    Since $L_A=D_A-A$ with $D_A=\dg\{A\one\}$, we have
    $$
        L_A-L_{W\odot P}=D_A-A-D_{W\odot P}-W\odot P\,.
    $$
    Therefore,
    \begin{align*}
        &\;\prob\lp \|L_{A}-L_{W\odot P}\|\geq t\rp\\
        = &\; \prob\lp \|D_A-A-D_{W\odot P}-W\odot P\|\geq t\rp\\
        \leq &\;\prob\lp \|D_A-D_{W\odot P}\|+\|A-W\odot P\|\geq t\rp\\
        \leq &\; \prob\lp \|D_A-D_{W\odot P}\|\geq t/2\rp+\prob\lp \|A-W\odot P\|\geq t/2\rp\,.
    \end{align*}
    We need to upper bound each term separately (We define $\Delta:=\max_{i}\sum_{j=1}^nP_{ij}W_{ij}^2$):
    
    \noindent
    \textbf{Upper bound for $\prob\lp \|D_A-D_{W\odot P}\|\geq t/2\rp$}: 
    
    \noindent
    For the first term, notice that both $D_A, D_{W\odot P}$ are diagonal, we have
    \begin{align*}
        &\;\prob\lp \|D_A-D_{W\odot P}\|\geq t/2\rp\\
        =&\; \prob\lp \max_{i}|[D_A]_{ii}-[D_{W\odot P}]_{ii}|\geq t/2\rp\\
        \leq &\; \sum_{i=1}^n\prob\lp |[D_A]_{ii}-[D_{W\odot P}]_{ii}|\geq \frac{t}{2}\rp\\
        =&\; \sum_{i=1}^n\prob\lp \lv \sum_{j}A_{ij}-\sum_{j}W_{ij}P_{ij}\rv\geq \frac{t}{2}\rp\\
       (\text{Lemma~\ref{lem_weighted_chernoff}} ) \leq &\; \sum_{i=1}^n2\exp\lp -\frac{t^2}{8\lp\sum_{j=1}^nP_{ij}W_{ij}^2+t/6\rp}\rp\\
        \leq &\; 2n\exp\lp -\frac{t^2}{8(\Delta+t/6)}\rp\,.
    \end{align*}
    
    \noindent
    \textbf{Upper bound for $\prob\lp \|A-W\odot P\|\geq t/2\rp$}:
    
    \noindent
    For the second term, let $e_i\in\mathbb{R}^n$ be the $i$-th column of the identity matrix $I_n$. Then
    $$
        A-W\odot P=\sum_{1\leq i\leq j\leq n}X_{ij}, \text{where } X_{ij}=\begin{cases}
            (A_{ij}-W_{ij}P_{ij})(e_ie_j^T+e_je_i^T), & i\neq j\\
            (A_{ij}-W_{ij}P_{ij})e_ie_i^T, & i=j
        \end{cases}\,.
    $$
    Notice that $X_{ij},1\leq i\leq j\leq n$ are mean-zero ($\expc[A_{ij}]=W_{ij}P_{ij}$), Hermitian, and we have $\|X_{ij}\|\leq 1, i\leq j\leq n$ almost surely. To apply Lemma~\ref{lem_mat_concentration}, we need to compute $\|\sum_{1\leq i\leq j\leq n}\expc[X_{ij}^2]\|$. Since
    $$
        \expc[X_{ij}^2]=\begin{cases}
            (1-P_{ij})P_{ij}W_{ij}^2 (e_ie_i^T+e_je_j^T), & i\neq j\\
            (1-P_{ij})P_{ij}W_{ij}^2 e_ie_i^T, & i=j
        \end{cases}\,,
    $$
    we have
    $$
        \sum_{1\leq i\leq j\leq n}\expc[X_{ij}^2]=\sum_{i=1}^n\lp \sum_{j=1}^n(1-P_{ij})P_{ij}W_{ij}^2\rp e_ie_i^T\,.
    $$
    Therefore $\sum_{1\leq i\leq j\leq n}\expc[X_{ij}^2]$ is a diagonal matrix with each diagonal entry upper bounded by
    $$
        \max_i\lp \sum_{j=1}^n(1-P_{ij})P_{ij}W_{ij}^2\rp\leq \max_i\sum_{j=1}^nP_{ij}W_{ij}^2:=\Delta\,.
    $$
    Invoke Lemma~\eqref{lem_mat_concentration} to obtain
    $$
        \prob\lp \|A-W\odot P\|\geq t/2\rp\leq 2n\exp \lp-\frac{t^2}{32\Delta+16t}\rp, \forall t\geq 0
    $$
    
    \noindent
    \textbf{Combining the two upper bounds}:
    
    Overall, we have
    \begin{align*}
        &\;\prob\lp \|L_{A}-L_{W\odot P}\|\geq t\rp\\
        \leq &\; \prob\lp \|D_A-D_{W\odot P}\|\geq t/2\rp+\prob\lp \|A-W\odot P\|\geq t/2\rp\\
        \leq &\; 2n\exp\lp -\frac{t^2}{8(\Delta+t/6)}\rp +2n\exp \lp-\frac{t^2}{32\Delta+16t}\rp\\
        \leq &\; 4n\exp\lp -\frac{t^2}{32\Delta+16t}\rp\,.
    \end{align*}
    Now set $t=8\sqrt{\Delta \log (4n/\delta)}$, the assumption $\Delta\geq 16(c+1)\log n$ implies $t\leq 2\Delta$. Therefore,
    $$
        4n\exp\lp -\frac{t^2}{32\Delta+16t}\rp\leq 4n\exp\lp -\frac{t^2}{64\Delta}\rp=\delta\,.
    $$
    This leads to exactly
    $$
        \prob\lp \|L_{A}-L_{W\odot P}\|\leq 8\sqrt{\Delta \log (4n/\delta)}\rp\geq 1-\delta\,.
    $$
    Recall that $\expc A=W\odot P$, we have the desired result.
\end{proof}
\begin{proof}[Proof of Theorem \ref{thm_spec_wsbm}]
    For the random matrix $A$ generated by the weighted stochastic block model $(\mathcal{I}_a,\mathcal{I}_b,p,q,w_p,w_q)$, we have
    $$\Delta=\max_i\sum_{j}P_{ij}W_{ij}^2=\max\{n_apw_p^2,n_bqw_q^2\}\leq npw_p^2\,,$$
    and
    $$
        \Delta =\max\{n_apw_p^2,n_bqw_q^2\} \geq n_{\min}qw_q^2\,.
    $$
    Therefore, for any $0<\delta<1$, pick $c>0$ and $n$ sufficiently large such that $n^{-c}\leq \delta$ and $\Delta\geq 16(c+1)\log n$. The latter is possible since $n_{\min}=\Omega(n)$, so that sufficient large $n$ has
    $$
        \Delta\geq  n_{\min} q w_q^2\geq 16(c+1)\log n\,.
    $$
    By Proposition \ref{prop_Lap_bd}, we have, with probability $1-\delta$,
    \be
        \|L_A-L_{\expc A}\|\leq 8\sqrt{\Delta \log(4n/\delta)}\leq 8w_p\sqrt{np\log(4n/\delta)}\,.\label{eq_pf_thm4_1}
    \ee
    Given the event \eqref{eq_pf_thm4_1}, by Weyl's inequality~\cite[Theorem 4.3.1]{Horn:2012:MA:2422911}, we have
    $$
       \lambda_3(L_{\expc A})-\lambda_3(L_A)\leq \|L_A-L_{\expc A}\|\leq 8w_p\sqrt{np\log(4n/\delta)}\,.
    $$
    Following the analysis in Appendix \ref{app_eig_lap}, we know $\lambda_3(L_{\expc A})=\min\{n_apw_p+n_bqw_q,n_bpw_p+n_aqw_q\}\geq n_{\min}(p+\gamma q)w_p$, then
    $$
        \lambda_3(L_A)\geq \lambda_3(L_{\expc A})-8w_p\sqrt{np\log(4n/\delta)}\geq n_{\min}(p+\gamma q)w_p-8w_p\sqrt{np\log(4n/\delta)}\,.
    $$
    This proves the first inequality. For the second inequality, by~[Theorem 2]\cite{Yu2014}, we have
    $$\|\sin \Theta(\hat{V},V)\|_F\leq \frac{2\sqrt{2}\|L_A-L_{\expc A}\|}{\lambda_3(L_{\expc A})-\lambda_2(L_{\expc A})}\,,$$
    where $\hat{V}=\bmt\frac{\one_n}{\sqrt{n}}& v_2(L_A)\emt$ and $V=\bmt\frac{\one_n}{\sqrt{n}}& v_2(L_{\expc A})\emt$, and the matrix $\sin \Theta(\hat{V},V)$ is a diagonal matrix of $\sqrt{1-\sigma_i^2(V^T\hat{V})}$. Notice that $\sigma_1(V^T\hat{V})=1$ and $\sigma_2(V^T\hat{V})=v_2^T(L_A)v_2(L_{\expc A})$, then
    $$\|\sin \Theta(\hat{V},V)\|_F=\sin(v_2(L_A),v_2(L_{\expc A}))\leq \frac{2\sqrt{2}\|L_A-L_{\expc A}\|}{\lambda_3(L_{\expc A})-\lambda_2(L_{\expc A})}=\frac{2\sqrt{2}\|L_A-L_{\expc A}\|}{n_{\min}w_p(p-\gamma q)}\leq \frac{16\sqrt{2}w_p\sqrt{np\log(4n/\delta)}}{n_{\min}w_p(p-\gamma q)}\,,$$
    which is exactly the desired inequality
    $$
        \sin(v_2(L_A),v_2(L_{\expc A}))\leq \frac{16\sqrt{2}}{w_p(p-\gamma q)}\sqrt{\frac{np\log(4n/\delta)}{n_{\min}^2}}\,.
    $$
\end{proof}